\documentclass[10pt,twocolumn,a4paper]{article}

\usepackage{graphicx,amsmath,amssymb,amsthm,mathabx,booktabs,txfonts,mathrsfs}
\date{}

\author{Yong Tan\\~\\
\emph{yongtan\_navigation@outlook.com}}
\title{A New Shortest Path Algorithm Generalized on Dynamic Graph for Commercial Intelligent Navigation for Transportation Management}

\theoremstyle{plain}
\newtheorem{prop}{Proposition}

\begin{document}
\maketitle 
\begin{abstract}Dynamic graph research is an essential subject in Computer Science. The shortest path problem is still the central in this field; moreover there is a variety of applications in practical projects. 

In this paper, we select the transportation activity as a paradigm to study. Of course there exists the most clear and appealing practical application over the intelligent navigation, also which is an important part in intelligent transportation system (ITS). Totally speaking, this activity is fairly simple and well-behavior under our consideration, for which, there is a lot of sophisticated theories and techniques to motivate our researches although this problem relates to interdisciplinary; so that we can be dedicated to computing aspect. Below, a host of practical problems will by the way be discussed in theory and empiricism including correctness, sampling, accuracy, compatibility, quick query and the application into other research; therein some have been scarcely argued intentionally or not in literatures before. Through these contents, we are able to have a smart system built which easy to incorporate other modules to realize a sophisticated industrial system.
\end{abstract}


\section{Introduction}
When we refer to a dynamic graph, at least, one can see that all the ingredients in network should be variable wherein includes every node, edge, weight and so forth. In general, people almost always refer all changes in this system to weight's change to simplify the control in order to reach the tractable; as it were, the weights totally conduct all activities by their value as well as realistic transportation: as our urban and the living traffic, the geographical information of urban may be often seen of relatively fixed through decades. The change is traffic state over time on every lane against that city and, everyone could be represented closely by a linear table whose members also change over time.

If this research over transportation arises into academic field, especially on Computer Science, it is well known for Cooke and Hasley in 1969\cite{1} firstly to propose the conception of \emph{Dynamic Graph}. Wherein, their contributions demonstrate these: (1) Define a dynamic graph, over there, every arc in graph is able to be connected with a weight table which is a discrete linear list used for reference, commonly called \emph{time-dependent} table. (2) There is likely to existence of a shortest path while given two nodes available as source or target in graph. (3) An adapted version of Bellman's method could be applied for the shortest path in such dynamic network.

In general, people say the shortest path is the \emph{fastest} path with the earliest \emph{arrival time} among all probable answers. This is our topic in this paper. In traffic engineering, there may be an important consensus for us is the realistic transportation is often presenting partially acyclic other than a totally chaos system. In light of this prerequisite, by the reference of those historical data over that road network, we may invent a virtual \emph{today} as being the predicted model to mimic the real today to make route oracle as guidance to conduct travelling. 

But people quickly found the Cooke's model is rather simple and needs to improve.

\subsection{Related Works}
For the sake of clarifying the correlation between the static and the dynamic, we briefly review the static heuristic method applied for the shortest paths, which kind is by now regarded as the highest efficiency in all analogues. The core idea was firstly proposed by Ford in 1956 reportedly\cite{2}, and the formula is as follows to characterize how to asymptotically make a path shorter, which is only used on a static graph.
\begin{align}\label{s1}
d(v)\coloneqq &d(u) + c(u,v)\nonumber\\
&\text{if}~d(v)< d(u) + c(u,v)~\text{or}~v~\text{not labeled.}\nonumber\\
&\text{otherwise no definition.}
\end{align}

Note that we have revised the original form, but no distortion if in contrast to Ford's original thought-way. In practical operation, we must get two variables recorded: the invariant $d(v)$ which represents the \emph{iterative} total length at node $v$; The dependent path $p$ is a queue consisting of many segments, whose piecewise from source to node $v$. Another may be node $u$, as soon as identified as available due to arc $(u,v)$ gets success to shorter path $p$, it will as unique predecessor to node $v$ and this arc should be maintained in a \emph{search tree} besides.

In the course of search, make node $v$ firstly assigned by function\eqref{s1}, we call \emph{labeling} method; on contrary to many times, we call \emph{correcting process}. Along with unceasingly correcting, the dependent variant $d(v)$ should be able to converge. As soon as all variants converge, upon that no necessary to be corrected, it means all target nodes thus obtain their respective optimality, due to this sufficient condition, the program thus terminates.

Once the static version the function\eqref{s1} adapted for dynamic instance, the way to get weights should be revised as follows in the nature of dynamic character, which reformation commenced from Cooke and Hasley\cite{1}.

\begin{align}\label{d1}
d(v)\coloneqq &d(u) + f(D, t,\ell): \ell=(u,v)\nonumber\\
&\text{if}~d(v)< d(u) + f(D, t,\ell)~\text{or}~v~\text{not labeled.}\nonumber\\
&\text{otherwise no definition.}
\end{align}

Is the function $f:\ell\rightarrow{R^{+}}$, defined by Cooke et al.; oneself called \emph{cost} function, which works over a discrete table $D=t_{0}, t_{0}+1,\ldots, t_{0}+M$ in which those weights thus settled. In that way, according to one table associated to one link, we have the dependent of traffic status and lane constituted. Meanwhile the parameter $t$ is only a positive integral for query in table $D$. And people generally say such table \emph{time-dependent} in order to emphasize variable \emph{weight} able to change over time, inversely, with departure time as being index, program could access the table for weight.

If there is an observation on two different versions of method, one can see they are possessed of \emph{non-decreasing} for iterative total-weight; besides, the different manner between them is the times of visiting table. The static is equivalent to precisely once beside which, in general, the dynamic maybe many. In fact, the more importance things are not in more than just to the times of visiting table, which pertinent to yielding weights.

Dean\cite{3} gave proofs for problem to be NP-hard can refer the premise to non-FIFO in time dependent model. The conclusion can be further referred to yielding weight is computed by cost function $f$ with the current state as necessary parameters, in term of functional definition\eqref{d1}. It is easy to comprehend that on account of search in table $D$ for that probable optimality, program would be only just in waiting until the enumeration for all items in table ends.

Apparently, Dean's work notices us that the Cookie and Halsey model is fairly vulnerable and susceptible of error.

In this aspect of transportation, a seminal work has been accomplished by Chabini and his students; they have done a lot of celebrated works especially to regulate the context of transportation problem. In \cite{4} , Chabini and Ganugapati intensively summarizes these premises for transportation problem: (1) \emph{Discrete} versus \emph{continuous}. (2) \emph{Fastest} versus \emph{minimum cost}. (3) \emph {One-to-all} versus \emph{all-to-one}. (4) \emph{Waiting allowable} versus \emph {unallowable}. (5) \emph{FIFO} (\emph{first in first out}) versus \emph{non-FIFO}.

Note that in practical problem or computing model, of those premises, however, some therein may be simultaneously co-existing in a system to comprise the context. As to transportation problem we study, the FIFO is generally a necessary condition that provides a chronological rule first-in-first-out: when two in a same order in a link and there is a delay between them besides, the latter never possibly overtakes the preceding one. If not so, the transportation problem should be NP-hard one, due to every latter perhaps is the optimal. In light of this clause, FIFO may regulate the \emph{waiting} at a node becomes a behavior to waste hours, whereas, sometimes, waiting is necessity such as delay for traffic signal administration, so we must have a consideration in treatment. 

Regarding algorithm, almost researchers argue a better way that to adapt the static wise for the dynamic circumstance. Dean even advises the search manner of One-to-All is the core algorithm basically for other ways\cite{3}. Chabini\cite{5}, Ahn and Shin\cite{6}, Dreyfus\cite{7} and Kaufman and Smith\cite{8} can be said to do a similar thing that proves applying adaptive version on dynamic instance is also suitable for shortest paths; as it does, to revise method for dynamic instance has been trivial issue unnecessarily in reiteration in this paper.

Turn to another dimension, listen to the voice that out from the industrial trades who work in engineering projects. For these demands they wish. Nannicini and Liberti summarizes those\cite{9} and wish the shortest path is the fastest; second, traffic data are updated at regular time interval; third, the computation on oracle have better to compute in real time and system can merely take a few millisecond to deal with million nodes.

Whereas in another research field, not that the fastest is the absolutely unique option for everybody the message we must need to capture. In \cite{10}, Gonzalez et al. study the security driving, such as that nobody wants driving through a high crime area. These mainly road navigation systems, up to date, however, do not supply any guideline to indicate security for drivers, they complained. To do that, by the wise of traffic data mining, they study a driving-security model in which, the data are squeezed into a discrete time table.

The two typical cases are telling us that for our computing model, is the best that the algorithmic framework should be possessed of compatibility, integrated and synergetic character besides; not only a simply powerful capacity of search.

\subsection{Our Contributions}
We outline a blueprint for intelligent transportation system via the way of generalizing our new algorithm over dynamic graphs. The significance is not only to introduce an algorithm which with an incredible performance and especially on dealing with large scale instance; but for some practical problems, which have not been mentioned or studied in previous literatures, we have resolved them in this paper with theoretical analysis. 

By our simulating experiments, we show off a prototype of empirical study to which we can clearly and easily measure the computing complexity and estimate the impacting factors with which, people must meet in practical operation. 

Finally, we give an instance of operation research to extend the application to urban logistics.

\subsection{Organization of Rest}
The rest of the lengths is organized as follows: Section 2 gives definitions, notations and introduction about our core algorithm we will revise. Section 3 defines the cost function into the time-continuous to resolve correctness; thus causes the discussion about sampling, accuracy, algorithmic overhead and the design on solution of experiment; eventually gives the discussion about experimental results. Section 4 is about transportation service, especially including deployment of computation, incorporating other modules, quick query for the fastest paths and extending the application in other field. The final section gives the summary.

\section{Preliminary}
\subsection{Domains and Notations}
A model to represent a transportation activity on real road networks one refers to can be composed of two domains the \emph{fixed} topology about geometrical information and the \emph{variant} of time domain to represent dynamic part: $\mathcal{T}=\{G, T\}$. The former item in set is a weighted-graph $G=(V,\tau, L)$, consisting of: (1) Set $V$ the collection of nodes; the scalar $n=\vert{V}\vert$. (2) The Greek character $\tau$ is referred to an arc-base set wherein in use of capital character $E$ rather than general character $m$ to denote this scalar $E=\vert\tau\vert$ besides. (3) As to each segment connecting two nodes in graph we call \emph{link} instead of the jargon \emph{arc} and, so to the \emph{edge}, herein we call \emph{lane} which at most containing a pair of links inversing mutually; wherein, (4) the link may be measured with the length function $\pi$, upon that having $\pi_{\in{L}}\colon\ell\rightarrow{R}^{+}$ for $\ell\in\tau$. 

All ingredients in graph $G$ may be thought of the fixed through the course of computation. Finally, of the character $m$, we use to denote the degree of \emph{outgoing} or \emph{ingoing} for a node in such link-base system.

Of the time domain, it can be represented of a set of time-dependent tables wherein the each form is $T_{\ell}=t_0, t_1, \ldots, M$. Anyway, we make every link $\ell$ in fixed domain associate with a time-dependent table $T_{\ell}$ to represent what happened over that link.

\subsection{Static Version Algorithm}
In this paper, due to the reason of the algorithm we will carry out adapted from the static version the Contest Algorithm (CA)\cite{11}, we have the necessary to comprehensively introduce this method in brevity. Totally speaking, this method CA possesses a set of well features, therein, which includes sub-linear runtime and a robust capacity that of oriented to data structure of graph rather than form, whose no susceptible of instance's constitution. 

The search manner of CA is \emph{One-to-All}, is the second property can be identified to our solution, as same as almost heuristic approaches going for the shortest paths. 

Alternatively, we can refer CA's success to employ the \emph{Best-First-Search} strategy to reduce the search scale in a large number so that the overhead could match the scalar of link quantity $O(E)$, which strategy is reportedly suggested by Gallo and Pallottino in 1988\cite{12}. By seeing the name, it is easy to understand the meaning as is Greedy-Tactic's kind: the system upholds the \emph{best} nodes that with minimal total weight preceding others and further correct them.

Idea is perfect, but there is needing an important module to support this heuristic kind; if not, the consequence would likely get worse than Bellman-Ford's method which in $O(nE)$. That is you must have an excellent \emph{priority heap}, which with a powerful performance to regulate those items to offer the best ones for search.

With an online database $LE$ as being such the heap, which derived from the Compound Binary Search Tree (CBST), in addition to do the routine works that already mentioned and the overhead of operation can be seen of the logarithm $\log{n}$ or constant; in particular, there is a special consideration to process a group of \emph{cousins} in constant runtime in tree, whose total weights are same one. 

To the algorithmic framework, the whole course may be seen to be composed of two phrases the \emph{labeling} and the \emph{correcting}. The antecedence is called Hybird Dijkstra's Method or HDM, which could be said to integrate the Dijkstra's algorithm and $A^{\ast}$ method. A set of \emph{Euclidean} shortest path maintained in search tree would come out in this phrase in runtime $O(E)$ that up with the scale of instance. 

Note that in this phrase, besides to get labeling over all target nodes for one-time to yield two lists loading a search tree and \emph{n}th total weights, which have been already mentioned in that discussion of function\eqref{s1}, there is another mechanism in effect to yield the \emph{third} list the \emph{origin} list, whose members would be likely to shorter others after executing HDM. Therefore this array will further become the input used in next phrase, the correcting. And it is thus maintained with data structure of database LE.

In correcting phrase, the overhead comes to sub-linear complexity that of $O(\lambda\log{n}+E)$ we call \emph{quasi-logarithm} and, the lower bound in $\Theta(n\log{n}+E)$ since the $LE$ can be seen as a random binary search tree in manipulation, so that the gross algorithmic overhead is $O(\lambda\log{n}+2E)$ to two phrases. 

We have got experiments to test CA on three typical instances: the grid instance, the random instance ($m\approx\log{n}$) and the complete instance ($m=n-1$). The results show the algorithmic performance higher than anyone before. Especially to large instance whose scale can up to mega level with million arcs, CA can still enable itself swimmingly to deal with.

The overhead of whole computing course is likewise parted into two shares: the \emph{Database Cost} or DC in polynomial $\Omega(n\log{n})$ for using database LE; the investigation on those links in $\Omega(E)$, called \emph{Search Cost} or SC.

That is easy to understand that proportion on two shares relates to scalars $n$ and $m$ in graph, which could be reflected by various forms. For example, to the weak connected instance with $n\ggg{m}$ like urban road network, the DC always takes the multitude share of the whole. But to the random graph with $m\approx\log{n}$ like internet network, both are closed.

In term of dynamic character, the changes of both will be distinct: beside DC the SC is indeed the \emph{functional} due to cost function, and the adding overhead than static version will be thus only stacked on the side of SC. 

If synchronize the practical runtime of DC and SC, we can use the former as benchmark to estimate the latter. To count the practical number, we can use the HDM in the SC's stead, because the \emph{labeling} and the \emph{correcting}, however, their orders of magnitude are equal of, as though they are in two different phrases.

\section{Time Continuous}
As to the account of the computing arrival time, it seems to be nature in those literatures before: the weight in time table $T$ is defined of through time in physical significance and the total weight is the iterative arrival time. This definition certainly is very simple and more familiar to manipulate query since table has been formed of a discrete list. 

Frankly speaking, we can refer a fetal error about arrival time to this interpretation above-mentioned. For example, consider two departure times $t_{i}$ and $t_{j}$ over a same link, both naturally used as two different indices directly for two thought times $w(t_{i})$ or $w (t_{j})$. 

Assume that the latter the departure time $t_{j}$ is delay to fore one then can be formulated: $\Delta{t}= t_{j} - t_{i}$ for $\Delta{t}>0$. Upon that we have an arrival time function $A= w(t) + t $ and get a difference $\Delta{A}=A_{j} - A_{i}$; and then 
$\Delta{A}=(w(t_{j}) + t_{j} )-(w(t_{i}) + t_{i})=\Delta{t} + \Delta{w}$. 

In addition, FIFO is a premise, so, if $\Delta{A}\leq 0$ then $\Delta{t}+\Delta{w}\leq 0$, which means the overtaking happens on these two vehicles. This is prohibited due to provision of FIFO and the subject we study is instantaneously to become a NP-hard problem besides.

In reality, the case is related: the traffic-through obtains a great improvement at that time point $t_{j}$, to the former, however at the departure time $t_{i}$, supposed of $ t_{j} - t_{i} =1$, there is still none the gain for oneself. So, the best economic decision for former is standing there than running forth. It is clearly ridiculous for this case to violate our common sense.

Apparently to the cause, there is no consideration on that time interval $\Delta{t} $ which relates to sampling frequency. However, we need to re-define this variable.

\subsection{Data Sampling}
As already mentioned, we questioned the correctness of arrival-time computing can be referred to definition of time interval. In the first time, we let a function $U$ to map the time interval to a real deterministic value: $U_{\in{T}}\colon \Delta{t}\rightarrow{R^{+}}$; and provide it equal to the weight involved to departure time $t$: $U_{\in{T}}(t_{i},\Delta{t_{i}})=w(t_{i})$. So for two indices $I_{j}$ and $I_{j+1}$ in time table $T$, they are mapping to two departure times $t_{i}$ and $t_{i+1}$, and then their difference is at least beyond the through time $w(I_{j})$. 

In this way, every time interval is no longer a fixed number, so the first benefit is to eliminate the case of overtaking absolutely. But yet, the collateral trouble emerges too because there is likely a departure time $t$ having $t_{i}<t<t_{i+1}$. For this case, besides meaning for accuracy to add relatively effort of calculation, anyway, the cost function has already become the \emph{time-continuous} rather than simple query in a discrete list.

In addition, it is not rigorous definition for the time interval either greater than or less than the through time. 

In the matter of fact, we have had no choice but to update the definition of weight. Our idea is to make table $T$ into a \emph{time-velocity} model: the time interval comes to a constant representing a fixed moment of time denoted by $\chi$ and $\chi > 0$; as it were, whose value is indeed to decide the traffic data sampling frequency.

And then we define the weight is \emph{mile} which the vehicle through as soon as possible during that fixed time interval $\chi$ rather than a time scalar before; we call \emph{velocity} denoted by $\upsilon$. In other words, weight is the highest speed vehicle could run

In this way, we have a gauge the scalar $\chi$ to measure a pair of departure times. When a departure time falls into an interval, it will imply a departure time should include two functions: one is for the index for table $T$ for velocity; another is for truncation for precise expression for real time over time interval.

It is naturally that our up-coming cost function will be a \emph{piecewise} function with a set of physical variables: there is a set of constants the length of link, encompass every one of theirs, we may take account of computing the through time by physical equation. There can be length function $\pi_{\in{L}}\colon\ell\rightarrow{R}^{+}$ for $\ell\in\tau$ or $\pi_{\in{T}}\colon\chi\rightarrow{R}^{+}$ thus to characterize a physical problem relevant to distance, velocity and time, upon that, we will have an inequality or equation as follows.

\begin{align}\label{ph1}
\pi_{\in{L}}(\ell) \leq &\sum^{\kappa}_{h=0}\upsilon_{I+h, ~T_{\ell}}\quad\text{ for }\kappa\geq{1}\nonumber\\
=&\pi_{\in{T}_\ell}(I+1, \kappa-1) + o( \pi(I+\kappa))+o(\pi(I)).
\end{align}

By means of this physical equation\eqref{ph1}, we may easily solve for the deterministic value of integral $\kappa$ ($\kappa\geq{0}$), and then to compute the large sum of through time by $\kappa\cdot\chi$, which over the link $\ell$ at the departure time $t$. At the right side of that equation, the function $o(f)$ is truncated function. In practical operation, this function is general used upon two segments in journey: one is beginning while the departure time falls in an interval, the rest journey might be unreachable, exceeding or jut reach the rightmost bound of the first interval; at the last segment of journey, we likewise have a precise truncating for rest length which shorter than the last velocity. 

In \cite{13}, Dean had ever to define such kind of piecewise function like ours; indeed the form Dean gave in that thesis is over complicated.

\subsection{Piecewise Function}
Therefore to visit time table with departure time $t$, it implies the departure time can be decomposed into two scalars by oneself over time interval $\chi$: the \emph{quotient} is index $I$ for table $T$; another is the \emph{remainder} for precise computation in first interval. If there is a rest mile, then continue to piece the mile to solve for variable $\kappa$. 

So the following pseudo code is a pattern by which, we just want to show the process to solve for variable $\kappa$ simply.

\begin{center}
@ Piecewise Function
\end{center}
\begin{flushleft}
\textbf{Input:} $t,\chi, \mathcal{T}$\\
$t=I\cdot\chi+{s}$; $e = \pi_{\in{L}}(\ell)$; $\kappa = 0$;~~\small{\emph{$\slash\slash$ Initializing.}} \\
~\\
\small{\emph{$\slash\slash$ computing the rest mile in first interval.}}\\
\textbf{If} $s > 0$ \textbf{Than} $e\coloneqq{e}-\pi(\chi-s)$; $I\coloneqq{I}+1$; \textbf{Else return};\\
~\\
\small{\emph{$\slash\slash$ piecing loop.}}\\
01. \textbf{Loop}: $e > \upsilon_{\ell}(I)$\\
02. ~~$e\coloneqq{e}-\upsilon_{\ell}(I)$; $\kappa\coloneqq\kappa + 1$; $I\coloneqq{I}+1$;\\
03. \textbf{End}\\
04. \textbf{If} $e>0$ \textbf{Than} $\varv = \upsilon_{\ell}(I)$;\\
\textbf{Output:} $\delta=\chi-s, \kappa, \varv, e$;
\end{flushleft}
~\newline
The exactly through-time is easy to calculate and denoted by $\delta$:
\begin{align*}
\delta&\coloneqq\delta+\kappa\cdot\chi + q+o(r) \\
&\textbf{s.t. } e=q\cdot\overline{\varv}+ r\quad\text{for }\overline{\varv}=\varv\slash\chi.\\
\end{align*}
The variable $r$ is final truncation. We call the whole process for through time \emph{Arithmetic Piecewise Function} or APF. 

Note that the property FIFO is still available while we provide velocity \emph{nonzero}. This proposition may be proven as follows.

\begin{prop}
To carry out APF as an essential calculation for through time, given two departure times $t$ and $t^{\prime}$ over link $\ell$; if $t \leq{t^{\prime}}$, then $t + \delta \leq{t^{\prime}} + \delta^{\prime}$.
\end{prop}

\begin{proof}
Due to conditions above-mentioned, two couples of variable, $t$ and $\delta$; and $t^{\prime}$ and $\delta^{\prime}$ may respectively comprises a finite space in time table $T_{\ell}$, denoted by $\rho$ or $\rho^{\prime}$ and, length function: $\pi(\rho)=\pi(\rho^{\prime})=\pi(\ell)$ holds true. Having an observation over that \emph{intersection} $\overline{\rho}$ of two spaces $\rho$ and $\rho^{\prime}$, the overlapping part of both has: (1) \emph{Empty}, implies $\rho$ ahead of $\rho^{\prime}$ on time-axis in table. (2) \emph{Nonempty}, but $\pi(\overline{\rho})<\pi(\ell)$; then $t+\delta$ in $\rho^{\prime}$. (3) $\pi(\overline{\rho})=\pi(\ell)$; two spaces overlap totally. Therefore $t + \delta \leq{t^{\prime}} + \delta^{\prime}$ holds.

\end{proof}

Turn then to algorithmic complexity, we so far can have $O(\lambda{n}\log{n}+E\mathcal{A}(T))$ to present the runtime complexity for dynamic version algorithm; which is a \emph{functional complexity} in term of piecewise function $\mathcal{A}(T)$, moreover, we use variable $\kappa$ as index to measure the overhead of APF module in computing process, which is the query time for time table $T$.

\subsection{Design on Experiment}
Our experiments need to consider some main problems on design: (1) how to deal with the case the function APF being time-continuous; (2) how to simplify the piecing computation and let the overhead at a lower level; (3) how to make the time table be endless.\\
~\newline
Therefore, we have a design as follows.
\begin{enumerate}
\item There are totally 46 grades in the fixed length set $L=(l_{k})_{k=1}^{46}\colon{l_{1}}=250; l_{46}=2500$ we make for chosen, which equivalent to mile (meters) where the fixed gradient is: $l_{i+1}-l_{i}=50$ for $0<i\leq{46}$. And let a couple of inversing links in a same lane with a same mile. 

\item Let velocity variable $\upsilon\in\mathbb{N}\slash\{0\}$ none-zero, moreover, we have 23 grades to represent this variable: $\varv=[10, 120 ]$ (kph) and by $5$ (kilometers) as the fixed gradient in series. This measure guarantees \emph{non-decreasing} property on APF. 

\item Let time interval $\chi = 10$ (second). The velocity vector $\varv$ thus has a counterpart at 10 second level: $\varv.10=[167, 2000]$ (meters); It is clearly that we can continue to constitute a set of similar counterparts: $\varv.i$ for ${i}>0$.

\item Upon that, we have thus got a Cartesian product $\Delta=\chi\times\varv$ by tow vectors $\chi$ into $\varv$, called \emph{Time Interval Domain} or TID and a query function $\delta_{\in\Delta}\colon i,j\rightarrow{R^{+}}$, in that way, function can return various value of velocity based on different grades of \emph{speed} and \emph{time interval}. For example, $\delta(1, 1)$ is 3 \emph{meter} over per \emph{second} which is running during 1 second at 10 kph. In our experiments, we let 1 second is the least precise for truncating. In term of our definition, the time interval is changeable according to our demands. In the way of query TID, program can approach to the final results by adding discrete numbers. Now in experiments, we initially let the TID be maintained in $10\times{23}$.

\item The time table $T$ thus merely carries the numbers of grade of velocity for which, we have the construct initialized by one-dimension list $T=[0, P)$. With conforming to the conception of \emph{infinite} time-table $T$, that says, we have to have got a finite table used to approximate an infinite one for our experiments. To do that, we provide when one visits the table with index $I$ as cursor, which must be limited in the scope of table: let index $I\coloneqq I-kP$ where $k=\text{max}(k)\in\mathbb{N}$ such that $0\leq{I}-kP<{P}$. In this way, the cursor $I$ is only to stay inside table $T$ whenever and whatever. 

\item In addition, because of the machine memory is less indeed (4G) for our experimental laptop, hence, we make all links share with a same time table $T$ instead of one to one, the pattern pertinent to link and time table in almost all literatures. Furthermore, in order to add challenge to imitate realistic situation, while a departure time is used to go for query velocity, there will be a pre-process on that parameter: we impose function to use the sum of departure time and length of link involved as being new departure time exclusively used for query. 

\item As to assigning time table $T$, we adopt two manners: the \emph{random} and the \emph{wave}, therefore, respectively called \emph{random network} or \emph{wave network}. In some sense, the former seems to make weights always vary in a volatile way over time; conversely, the change on wave may be thought of smooth and mild; the varying is moderate either for upwards or for downwards.

\item At the end, we use \emph{square} grid instances where both row and column with same quantity of nodes denoted by $k$ and the instance denoted by $\varg.k$; actually having $k=\sqrt{n}$. For example, $\varg.30$ is the instance with 900 nodes ($30\times 30$) and the number of links is round 3600.

\end{enumerate}

In experiments for each instance, we let program randomly select a node in graph as \emph{source} with a stochastic natural number as being initial departure time. The following table.1 (Appendix) exhibits the primary resulting data on a set of instances, wherein the max scale up to ten millions.

Herein, we call this algorithm \emph{Dynamic Contest Algorithm} or DCA adapted from static version CA (the source code has been hosting in that site: \emph{https:\text{//}snatchagiant.github.io\text{/}Shortest-paths-algorithm-over-dynamic-graph-for-transportation-management}, and our experiments likewise based on those codes).

\subsection{Interpretations of Experiments}
The results yielded from our experiments are surprisingly, which include three aspects: the instance scale, the overhead of APF and the search efficiency.

Firstly we estimate two types of network: The random is faster than the wave for 20\% in executing time, correspondingly to times of querying time-table $T$, the index of the \emph{average times of querying} (ATQ) is still to win the wave for less than 27\% in general. 

The other experimental data, especially relates to transportation, however, two types are unexpectedly closed. Actually, the comprehensive data show more interesting as follows.\\
~\newline
\textbf{Instance Scale. }These experimental instances are listed in the leftmost column of table.1. We may have a conception about the instance scale. Synchronize the scale of ours and reality, it may be to say: of the graph $\varg.50$, we can view as an \emph{average urban} with 2500 nodes; the $\varg.100$ with 10000 nodes more than Montreal (Canada) that with 7000 nodes\cite{5}, which might into a \emph{metropolis}; a \emph{state} scale can match the $\varg.1k$ with 1M nodes; the last one $\varg.3.5k$, we can say \emph{continental} magnitude with more than 12 million nodes and 48 million links.\\
~\newline
\textbf{Runtime. }The runtime in the second column, at least they show our program can halt and get a pretty score. There is a typical case can be referred in \cite{14}: with reference to the report from Goldberg et al., they took \emph{years} on testing an algorithm whose complexity in $O(nm)$ (where $m$ equal of our $E$) whose counterpart is Bellman-Ford's method. In that experimental instance, there were more than 29 million nodes and at least 70 million arcs. Note that, the record is thus to come out upon a static instance the real road network of Northern American Continent with real data.

In contrast to Goldberg's score, ours took less than 109 \emph{second} to achieve the computation on a 12-million-nodes grid instance. We yet cannot find out any more similar records in literatures, more correctly, for the dynamic shortest path computation within such large scale or more.\\
~\newline
\textbf{Hardware and Software. }To experimental equipments including hardware and software, simply compare Goldberg's and ours: (1) Goldberg's work station with core frequency 2.4G versus our laptop with 1.8G; and the machine memory also win ours by 16G versus 4G; (2) both C++ as compiler, but there is a bit difference between both OSs the Windows Server 2000 and the Windows 10 (home edition version). Of course, the memory on our laptop finally restrains our testing on larger instance.\\
~\newline
\textbf{The Search Efficiency. }If there is a pair of physical places as \emph{source} and \emph{target} in instance we call \emph{s-t-pair}. In the nature of things, the one-to-all search tree contains $n-1$ links, meanwhile, which rich in \emph{s-t-pair}s. Since any sub-path in a shortest path is also the same kind, therefore these sub-paths should be valuable able to benefit all-to-all. 

The fourth column in table.1 illustrates a well quality of our model that the average cost on each \emph{s-t-pair} would decrease along with the increment of scale as batch process in manufactory, the inverse relation of unit cost and scale.

The good news is at the rightmost column in table.1, which exhibits the average querying time on each links can be said to approach to a stable value against the time-continuous context of APF. Furthermore the table.2 (appendix) gives the corresponding percentages about two shares of DC and SC into the gross overhead. Whereas, these experimental data are on appearance to render the DC share would increase and get to the chief overhead of computation against the instance swelling. 

From a different angle, the case proves the piecewise function tends to present a constant factor into variable $E$ in complexity polynomial, which certainly reflects those features on velocity including distribution and the range of value in dynamic domain. Apparently, this number can be an index to reflect the traffic state in some extent.\\
~\newline
\textbf{Sampling. }Recall our model of time-velocity to represent a time-continuous, in this way, we make scalar $\chi$ changeable of the practical demands. It may be imagined to lower this scalar, and due to the fixed lane and the length is the upper bound for piecing, the index ATQ should be inevitable of increment and, the all adding effort will be stacked at SC side.

The table.3 shows the trials on $\varg.1000$ instance which let $\chi\in\{10, 9,8,7, 6,5,4,3\}$ respectively. These resulting data support our prediction totally. Here, our trials just elaborate two points: one is the flexible for gathering traffic data; two is we can study the trade-off about the sampling interval, accuracy and computing overhead via this way of simulation.

\section{Transportation}
We turn to the transportation field. The first content to be discussed is those general data about average length, average velocity and average through time which been listing in table.4 (appendix).

On appearance, those results give us an overall impression: despites of the scale of instance change, all average through times would around the number of 60 second; even to on two different networks the random and the wave. Of course, as our already discussions, the case is able to relate the data constituent and distribution. But we will not go to the issue further. Our interest will go on to the algorithmic framework, deployment of computation, service model and applying development.

\subsection{Synthesis}
If we decompose DCA frame work, we can have a comprehensive review: (1) The online database LE is a functional module and its job is just to serve for sorting items, and then it is able to be demerged off process as an independent functional module. (2) Alternatively there are twice to invoke a same labeling module in whole process. We likewise can abstract it function as an estimation system entirely.

To this system, we have an innovation relating to embed a set of modules into this system motivated to build an integrated one. For example, to a target $v$, assume there is a group of links $(\ast, v)\subseteq\tau$ towards to node $v$, in which, they all satisfy the condition of fastest path, and then program can randomly pick up one as result. One can say these iterative arrival times at $v$ are equal of and minimal, nevertheless a set of scalars on them is different each another, which including the departure time, the length of link, the velocity and the through time at all. 

We now select length function to enumerate the minimal length of link involved, as follows.

\begin{equation}\label{m1}
d.v\coloneqq{u}\quad\textbf{s.t.}~\pi(u,v)=\text{min }\pi(\ast, v).
\end{equation}
It is clearly that the length function\eqref{m1} \emph{behind} function APF in executing order of procedure. By intuitive sense, the final results may be said of the \emph{shortest-lane} in all answers of the fastest path, conversely, swap their several positions, the result would be the fastest path underlying the shortest lane.

One is quickly to refer to the halting problem. It is no problem to halt because two functions are \emph{convergent}. In contrast to this feature, we update length function to enumerate the \emph{longest} link. The same case of length function behind APF, the resulting context seems to no problem: the longest underlying the fastest. But to the case of the ahead of APF, we cannot suggest anything for physical significance.

Although to gather two modules with different features, yet, finally the APF imposes the program's halting. Here we are not going to continue the further study. We are only just notice readers: as though we got a no bad result this time, to incorporate other modules, you must have a consideration on functional feature, at least the convergent or the divergent.

On the other hand, we have exhibited a good frame work in our solution easy to incorporate others modules, which owes to the feature of the oriented framework.

\subsection{Waiting Allowable}
There is a kind of problem we must meet with, that is \emph{waiting allowable}, more correctly, it allows vehicle in our model able to wait at nodes or link for a period of time. Through setting velocity $\upsilon$ equal to zero or a tiny number, the aim is easy to reach; but we understand that besides adding the effort in piecing process in APF to represent the vehicle ever in waiting, either at node or at link, another problem is more serious which might lead to our problem becomes a NP-hard one as already mentioned. 

On the other hand, the challenge is this problem is necessity as a realistic issue pertinent to the traffic administration by signal, vehicles congestion in lane, road briefly closure and something else. Likewise, allow a long waiting time may cause a bigger risk to our model which can deteriorate the quality of computation; as it were, program might almost equivalent to non-halting; as to this point, we can refer cause to piecing thread: in term of APF definition, the process is able to be checked for waiting piecing thread ending, in this way, a set of scalars comes out for next phrase. In the nature of cases, it would still get worse in particular of the case happening which horrible congestion is destroying the whole traffic.

In \cite{3}, Dean translates the co-existing relationship of waiting-allowable and FIFO: the key point is the arrival time is possessed of \emph{non-decreasing} property that is $a_{\ell}(t)\geq{t}$, otherwise NP-hard. This remark implies whatever the latter does, it would at most beside the former but never to overtake. 

Therefore we need to prove: nevertheless velocity be set to zero to represent the waiting allowable, our model would be still to respect the provision of FIFO. To prove that, the key point is to survey the case of two vehicles via the same order whereas there is a delay between them.

\begin{prop}
We let the velocity in time table $T$ might be zero, the FIFO law is still available while carrying out APF.
\end{prop}

\begin{proof}
Following proposition 1, there can be a complement $\overline{\rho}$ by space $\rho$ \emph{minus} $\rho^{\prime}$. Obviously, if assume their arrival times equal of, then all velocities in the complement $\overline{\rho}$ are zero or trivial able to be truncated down; on contrary to one in $\overline{\rho}$ no zero, then leads to two lengths $\pi(\rho) > \pi(\rho^{\prime})$, the contradiction to the premise for through a same link. The FIFO clause is still available.

\end{proof}

Through our proof, we understand the waiting conception is partially available just due to overtaking short; meanwhile, we must attend the quantity of zero-velocity in piecing thread lest to damage quality of our computation.

\subsection{Summary}
Review those discussions we had made, our computing model has been possessed of four properties, which are the FIFO; One-to-All; Time-Continuous; Fastest Path and Partial Waiting-Allowable. These four construct the context underlies our study on transportation.

Our model has at least realized these functions: There is a powerful performance able to deal with large instance that up to continental level is \emph{first}. \emph{Second}, traffic sampling frequency is alterable up which, people can be flexibly to set according to their demands rather than before, the error and the rigid. \emph{Third}, algorithmic frame work is enough simple or clear to be decomposed for deployment to meet user's want. \emph{Fourth} is with a robust character to incorporate other estimation modules if feasible. \emph{Fifth}, allow the waiting case occurs in model to imitate real-time case but it may incur the risk of problem of non-halting for piecewise function.

Therefore, we can adopt a wise of truncating the chronological series in table $T$ to prevent from non-halting: provide a large number $M$ as being upper bound for query length in table $T$. Within the block data, if APF cannot finish computing, then judge the link being interruption and stop the search over this link. In this way, the overhead of APF can be seen of $O(M)$ at worst case.

On the other hand, this mechanism as well benefits us to save time resource: to distributed computation, the time table $T$ can be hosted in a server as a database shared with all processes. If a process need many times of query for that database over a link, in this way, it maybe leads to a lot of queries thrown in a long queue in waiting. That is not worth to the waiting time.

\subsection{Quick Query}
Another problem arises that whereas we tackle the shortest path problem and optimize our solution as possible as we can, but as it were, these capacities is still relatively limit and hard to accomplish much more. 

A report from \cite{15}, where Sanders and Schultes implemented the Dijkstra's method on a \emph{static} network of Western Europe that with 18 million nodes wherein those fixed weights assigned with average travel time, their machine took less than 9 seconds whose hardware with Processor 270 (A$\ast$D) clocked at 2.0 GHz with 4 GB main memory; procedure runs in SuSE Linux 10.0 OS and complied in use of C++. To contrast with our trials, implement a HDM on a \emph{dynamic} instance that with 12 million nodes, which took less than 24 seconds. Although there is no qualitatively difference between both of Sanders and ours, yet it is a sweet hour to quick query for fastest path within a continental extent. 

In some sense, Dijkstra's method and ours has a common defect which is the manner of search the one-to-all. It seems to overuse resources to deal with a query for shortest path through sparing all resources even to achieve a little task; something seems not a bit of necessity.

The case is easy to be translated to a contradiction between users and ours: user just wants to a tailored service to submit a \emph{s-t-pair} to system only for a path; instead, our solution resembles the manufactory whose way inclines to mass production, a cluster service.

The conflict is not easy to reconcile, so that people want to have a consideration benefiting both. Well, the conception of \emph{overlay} graph has thus been argued and advocated by some researchers. The core idea is to extract the useful information and concentrate them in a set, more correctly, a cutting graph in rather small scale enough to support online query; as it were, we can view this set as a buffer among users and us.

Here we refer to two typical works: Goldberg\cite{14} on $A^{\ast}$ method with landmark and the highway hierarchies of Sanders\cite{15}. They call the wise of extracting overlay from original graph \emph{pre-computation}.

Of course, this kind of wises is scarcely to apply on dynamic network at least because of two points: (1) In general, the mechanism are extremely rigid so as to susceptible of error due to network state change even if a little bit, people call the case \emph{data perturbing}. For example of the \emph{landmark} method, it works based on finding shortcut depending on inequality like Euclidean planar. In \cite{9}, it reports a small number of motorway changes, 95\% queries could be slowly in severe degeneration. That means needing to renew to compute from scratch. (2) It is about the core algorithm the \emph{bidirectional} Dijkstra's method, which cannot deal with the dynamic states well. In that instance, the unique knowable state is at the source other than the static which every node is known, thus, this algorithm may not be employed because target endpoint's state is unknown

Therefore, of such kinds, the idea succeeds but core algorithm fails, we call \emph{Static Pre-Computation} or SPC.

To our solution, it is easy to understand that by implementing DCA, we have a search tree yielded. Hence, we make an overlay graph actually becomes a data mining. Those objects are \emph{s-t-pair}s in that tree and the existed shortest paths.

We hereby roughly partition the service model for quick query for shortest path into two interfaces\textemdash the \emph{online interface} and the \emph{offline interface}. The latter's job is pre-computing. And we will encompass this model to stretch our idea and discussion. 

\subsection{All-to-All Model}
As already mentioned, our goal is to supply the shortest path for every probable query that of a \emph{s-t-pair} at any feasible departure time. It is easy to understand we are going to finish the toil of computing \emph{all-to-all} for every \emph{s-t-pair} without fail. Against these resulting search trees, it seems like a great easy job but to reckon the overhead of coverage.

We denote the overhead by $O(F)$ to implement DCA for once. If we have every node in graph as source for once, the all-to-all coverage is surely to accomplish but the complexity should equal to implementing a Bellman-Ford's method for once. For example of $\varg.3.5k$, refer the practical runtime in table.1; which shows taking 1.5 \emph{minute} to execute DCA on that scale for one time. We can reckon the executing time for all-to-all coverage in use of our laptop. The result is easy to obtained that up to four \emph{decades} by $1.5\cdot{3500^2}$. This is why Goldberg reported \emph{years} to implement a $O(nE)$ algorithm\cite{14}. 

Without questions, in that way, there is a great deal of repetitive paths yielded in that covering course. Here we firstly reckon the lower bound, upon that, we need to count the quantity of \emph{s-t-pair} in a search tree.

We here suggest a triangle-area wise to solve for that number as through it is approximate: we can see the \emph{depth} of tree as \emph{height} of a triangle where the source is \emph{vertex}; the \emph{area} is known of the nodes amount $n$; finally the \emph{bottom} can be the amount of terminals in the tree. When we solve for the bottom, through this number, we can approximately learn the quantity of paths in tree which from source to all terminals. 

\textbf{Example. }To an instance $\varg.k$, we suppose source at a corner on instance. The \emph{shortcut} length (the amount of links) to the farthest target node might be $2k-1$, can as being height. The bottom is obtained of $k=\sqrt{n}$ swiftly by double area over height: $2n\slash{(2k-1)}$.

Furthermore, we thus have a Stirling number $k\cdot{C}_{2k}^2$ to reckon. The outcome is $\sqrt{n}n$. The experimental data on third column in table.1 supports our analysis. Thus, we have the lower bound of all-to-all overhead: $\Theta(\sqrt{n}F)$; the $k$ times to implement DCA.

Of course, if one wants to touch this bound, he needs to find a set of \emph{proper} nodes as being source. Certainly, we so far have not any non-trivially theorem to support this action. 

To the upper bound, it can be predicted that it seems to $O(nF)$ that means each node being source for once; the way is impractical. Here we have a \emph{conjecture}: every node in circumference of graph as sources for once, we have an all-to-all coverage; called \emph{Knitting Conjecture}. 

To a square grid instance, the upper bound could be $O(4\sqrt{n}F)$ only if knitting conjecture holds. Then on instance $\varg.3.5k$, the runtime could be 6000 to 24000 \emph{minute} or 100 to 400 \emph{hour} of wall clock time on a laptop.

To a metropolis scale $\varg.100$, score would be in the scope from 6.2 to 25 \emph{second} on one laptop for knitting conjecture; even to every node as source, the score would be less than 10 \emph{minute} too. As things stand, one can say that covering computation is feasible at urban level.

\subsection{Overlay Zone}
As things stand, we seem to achieve the all-to-all coverage and set up database to deposit a set of shortest paths for every \emph{s-t-pair} in graph. If so, it is easy to assess the memory space in $O(n^3)$ to store a snapshot of an instantaneous traffic state. 

In the nature of cases, to those of hot-query points, this idea can be said of no too bad. Actually, the operations will become rather rigidly and wearing resources: to a same \emph{s-t-pair}, for which, there could be many shortest paths corresponding to different time points, a \emph{discrete} time-dependent list. Apparently, among these paths, it is possible for majority in each are totally identical; but database yet has to see them as being the different.

Instead, if merge all nodes in some paths as a cluster or a \emph{cutting} graph out of the original, and if it is \emph{rather} small, however, the online implementing DCA within that scale and the overhead is completely affordable, not to impact the experience of user. We call this overlay graph \emph{Hot Zone} denoted by $Z$.

The benefit for this zone is clearly that it represents time-continuous by which, the query can be submitted at any time point in a feasible period. Thus, we get a flexible measure to deal with the problem which is the traffic state varies over time. Servers need not store too more details.

\textbf{Example. }Now we set the \emph{size} of $\vert{Z}\vert = 8\sqrt{n}$. For instance $\varg.3.5k$, the shortest fastest path may contain about 7$k$ nodes. According to our data in table.4 about the average through time over a link, the journey on so long path should need a vehicle tirelessly to run in about 6 days to cover ground and the maximal precise deviation is 45 \emph{minute} for \emph{second} as the minimal unit of time.

In term of our provision about the max size of zone, we can have a max one with $28k$ nodes in $\varg.3.5k$. If a $10\times 2800$ instance for equivalent simulation, we have less than 250 \emph{ms} token in our laptop for implementing DCA on that zone for once. Turn to the scale of metropolis the $\varg.100$, the equivalent max-zone would be $\varg.30$ and we have the runtime within \emph{microsecond}. 

Certainly, there is a host of small zones in a larger one with many different departure times and arrival times. What is the exact and the cheap way to extract them, we leave the problem down to readers. At the end, to a routine work, the search complexity for zone model could be in $O(\lambda\tilde{n}\log\tilde{n}+2M\tilde{n})$ for $\tilde{n} = \sqrt{cn}$ and $c\geq 1$.

\subsection{Development}
In this section, we will discuss application for other works; for example the emergency of congestion in road network. If we consider the detour almost always in vicinity of those roads that in zone and severance, then by the zone swelling, we have another shortcut to guide driver out from the dilemma.

An important application comes from the field of urban logistics. This problem even seems so far to have no formal name in literatures, which is about how to find out a \emph{best} departure time in a given period of clock time to launch a journey at a source such that we have the through time minimized to reach the destination lest to meet with the rush hours. We here call this problem \emph{Best Opportunity Problem} or BOP. 

In general case\cite{16}, the typical measure is to employ Bellman-Ford's method in repetition for search at many time points; finally by sorting those outcomes, the optimality will be thus obtained. The contributions Ding and his team have done are working $A^{\ast}$ method as the core algorithm instead of Bellman-Ford's.

However, the theoretical analysis on those wises is beyond our scope of this paper. We just only suggest by the way of our offline computation, the problem may be obtained to resolve: (1) A \emph{s-t-pair} and a time period user gave is known; correspondingly, offline interface could supply a set of zones which represent the traffic state about that \emph{pair} during a period that at least contains the given one; so the search would be limited in this scope that implies: the time point and zone for the possibly fastest path. (3) As already mention, many times to implementing DCA on such smaller scale cannot make massive overhead, surely great lower than Bellman-Ford's for one time over the all resources from scratch. (4) The computing time interval could be arbitrary absolutely according to your demands just because the zone is time-continuous; of course, you need make a trade-off for multi aspects.

If deal with a host of queries for BOP, then our cluster service would reduce the overhead over each query much more than the current, nevertheless the emphasis is partially to solve BOP because by the means of offline computation.

\section{Conclusion}
In a word, we give a relatively comprehensive discussion on transportation problem by encompassing a powerful algorithm on dynamic graph. By analysis of computation complexity and the experimental simulation, we succeed to discuss a host of problems, therein including some practical ones which had been ever ignored before, and new topics.

In the truth of matter, we only give the principles for computing model or approximate data. But, we can believe our suggestions in this paper may inspire people to think about how to build a sophisticated intelligent transportation system, which is integrated, exact, modular and the flexible oriented framework.

\newpage
\begin{center}
\textbf{\Large{Appendix}}
\end{center}

\begin{center}
Table.1 \small{\textbf{Performance (wave)}}\\
~\\
\renewcommand{\arraystretch}{1.2}
\begin{tabular}{l|rrrr}
\hline
\footnotesize{\textbf{Inst.}}&\footnotesize{T(ms)}&\footnotesize{\emph{s-t-pair}s(S)}&\footnotesize{T/S($\mu{s}$)}&\footnotesize{ATQ}\\
\hline
$\varg.50$&15&$1.13\cdot 10^5$&0.132&17.1\\
$\varg.100$&62&$5.77\cdot 10^5$&0.107&16.8\\
$\varg.400$&1\,031&$5.1\cdot 10^7$&0.020&17.2\\
$\varg.1k$&7\,031&$8.9\cdot 10^8$&0.008&17.5\\
$\varg.2k$&32\,141&$7.6\cdot 10^9$&0.004&17.4\\
$\varg.3.5k$&108\,484&$3.6\cdot 10^{10}$&0.003&17.2\\
\hline
\end{tabular}
\end{center}
\footnotesize{\textbf{REMARK: }\emph{Item T is the practical runtime counted by millisecond (ms); the second ratio of \emph{s-t-pair}s amount into T in use of microsecond ($\mu{s}$) as physical unit. ATQ is average time of query by total query number over arcs amount.}}
~\newline
\begin{center}
Table.2 \small{\textbf{Overhead (wave)}}\\
~\\
\renewcommand{\arraystretch}{1.2}
\begin{tabular}{l|rrrr}
\hline
\footnotesize{\textbf{Inst.}}&\footnotesize{SC}&\footnotesize{DC}&\footnotesize{DC$\slash$SC}&HDM(ms)\\
\hline
$\varg.50$&0\%&100\%&null&0\\ 
$\varg.100$&48.4\% &51.61\%&1.07&15\\
$\varg.400$&51.4\%&48.59\%&0.95&265\\ 
$\varg.1k$&48.4\%&51.56\%&1.06&1\,703\\
$\varg.2k$&46.0\%&54.01\%&1.17&7\,391\\
$\varg.3.5k$&43.9\%&56.10\%&1.28&23\,812\\
\hline
\end{tabular}
\end{center}
\footnotesize{\textbf{REMARK: }\emph{Item SC and item DC respectively represents percentage of share in total overhead. HDM is the number of milliseconds for executing HDM.}}

\begin{center}
Table.3 \small{\textbf{Sampling Frequency and Overhead (wave)}}\\
~\\
\renewcommand{\arraystretch}{1.2}
\begin{tabular}{l|rrrrr}
\hline
\footnotesize{$\chi=$}&\footnotesize{ATQ}&\footnotesize{SC}&\footnotesize{DC}&\footnotesize{DC$\slash$SC}&HDM(ms)\\
\hline
10&17.2&48.4\%&51.64\%&1.07&1\,719\\
9&19.0&49.3\%&50.75\%&1.03&1\,797\\ 
8&21.4&51.7\% &48.25\%&0.93&1\,969\\
7&23.5&53.1\% &46.91\%&0.88&2\,078\\
6&27.1&56.8\%&43.24\%&0.76&2\,297\\ 
5&31.6&59.0\%&41.02\%&0.70&2\,516\\ 
4&38.4&62.4\%&37.57\%&0.60&2\,907\\
3&48.7&65.6\%&34.41\%&0.52&3\,500\\ 
\hline
\end{tabular}
\end{center}
\footnotesize{\textbf{REMARK: }\emph{The scores is out from the experiments on $\varg.1k$ instance.}}

\begin{center}
Table.4 \small{\textbf{Transportation Data (wave)} }\\
~\\
\renewcommand{\arraystretch}{1.2}
\begin{tabular}{l|rrrr}
\hline
\footnotesize{\textbf{Inst.}}&\footnotesize{$\widebar{L}$ (meter)}&\footnotesize{$\widebar{V}$ (kph)}&\footnotesize{TT ($\widebar{L}/\widebar{V}$)(sec.)}\\
\hline
$\varg.50$&1167.93&67.65&62.15\\ 
$\varg.100$&1119.94&68.19&59.80\\ 
$\varg.400$&1132.64&67.14&60.05\\ 
$\varg.1000$&1122.73&66.17&60.08\\ 
$\varg.3500$&1471.61&88.43&59.91\\
\hline
\end{tabular}
\end{center}
\footnotesize{\textbf{REMARK: }\emph{The all items are average numbers over link\textemdash $\widebar{L}$ is of length; $\widebar{V}$ is of velocity; TT is of through time.}}
\begin{center}
\includegraphics{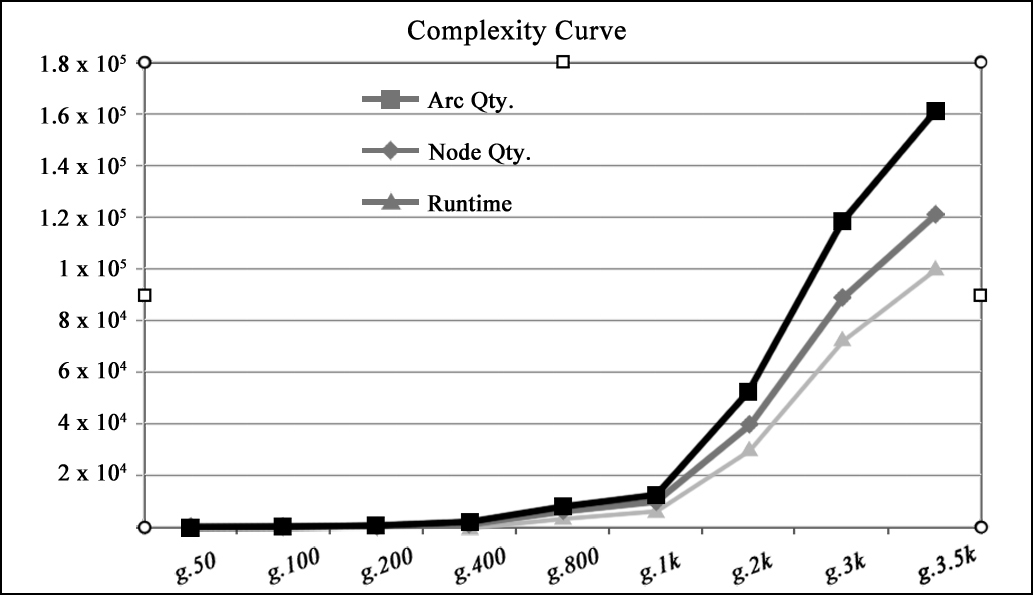} \\
~\newline
Figure.1\end{center}
Figure.1 above is presenting a triple of variants: nodes number, links number and runtime (counted by wall clock the millisecond). 

\begin{center}
\includegraphics{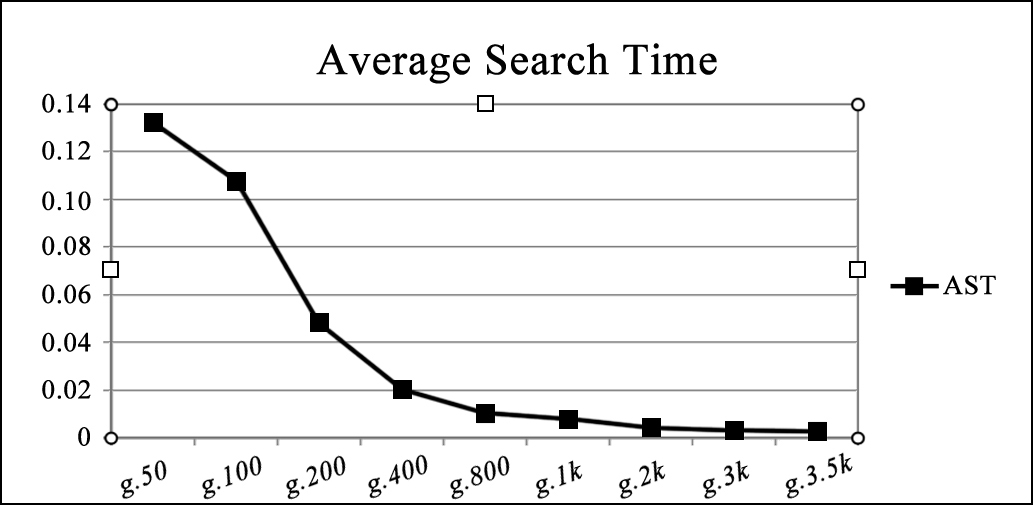} \\
~\newline
Figure.2\end{center}
Figure.2 is illustrates the change of search efficiency if we measure it with average search time on \emph{s-t-pair}.

\begin{center}
\includegraphics{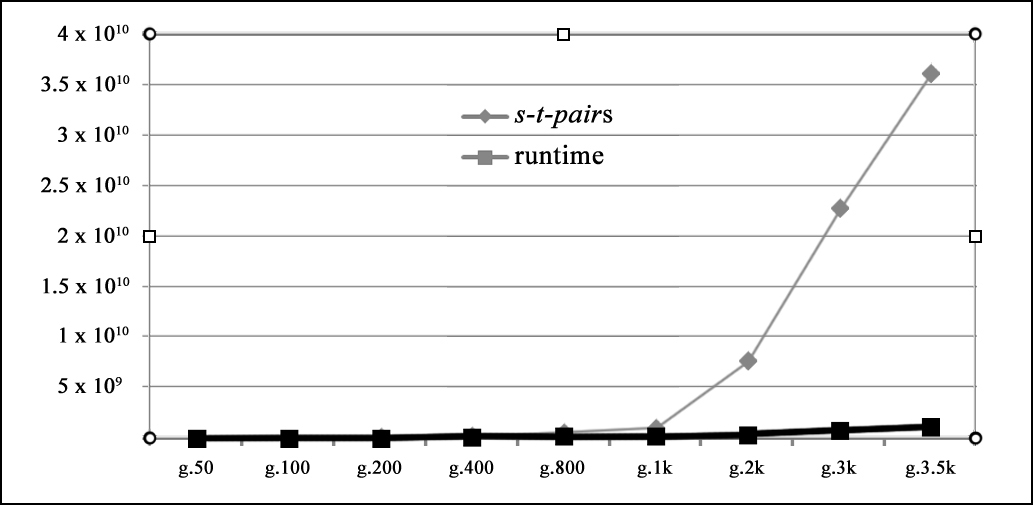} \\
~\newline
Figure.3\end{center}

We let numbers of runtime to multiply $10^4$ to show the tendency of \emph{s-t-pair} scale and runtime. They are respectively represented by the square root of $n$ and logarithm. Apparently the increment of square root is fastest than logarithm a lot. 

\end{document}